\documentclass{article}
\usepackage{amssymb,amsfonts,amsbsy,amsthm}

\begin{document}

\newtheorem{theorem}{Theorem}
\newtheorem{algorithm}{Algorithm}
\newtheorem{corollary}{Corollary}
\newtheorem{definition}{Definition}
\newtheorem{example}{Example}
\newtheorem{lemma}{Lemma}

\title{Unirational Fields of Transcendence Degree One and Functional Decomposition\footnote{This work is partially supported by Spanish DGES Grant Project PB97--0346}}

\author{Jaime Gutierrez \and Rosario Rubio \and David Sevilla}

\date{}

\maketitle

\begin{abstract}
In this paper we present an algorithm to compute all  unirational
fields of transcendence degree one containing a given finite set
of multivariate rational functions. In particular, we provide an
algorithm to decompose a multivariate rational function $f$ of the
form $f=g(h)$, where $g$ is a univariate rational function and $h$
a multivariate one.

\end{abstract}
\section{Introduction}
Let $\mathbb{K}$ be an arbitrary field and
$\mathbb{K}(\mathbf{x})=\mathbb{K}(x_1,\ldots,x_n)$ be the
rational function field in the variables ${\mathbf{x}}
=(x_{1},\ldots, x_{n})$. A unirational field over $\mathbb{K}$ is
an intermediate field $\mathbb{F}$ between $\mathbb{K}$ and
$\mathbb{K}(\mathbf{x})$. We know that any unirational field is
finitely generated over $\mathbb{K}$ (see \cite{Nag93}). In the
following whenever we talk about ``computing an intermediate
field'' we mean that such finite set of generators is to be
calculated. The problem of finding unirational fields is a
classical one. In this paper we are looking for unirational fields
$\mathbb{F}$ over $\mathbb{K}$ of transcendence degree one over
$\mathbb{K}$, $\mathrm{tr.deg}(\mathbb{F}/\mathbb{K})=1$.
\par
In the univariate case, the problem can be stated as follows:
given univariate rational functions $f_1,\ldots,f_m\in
\mathbb{\mathbb{K}}(y)$, we wish to know if there exists a proper
intermediate field $\mathbb{F}$ such that
$\mathbb{K}(f_1,\ldots,f_m)\subset \mathbb{F}\subset
\mathbb{K}(y)$; and in the affirmative case, to compute it.  By
the classical L\"{u}roth theorem (see \cite{vdW66}) the problem is
divided in two parts: first to compute $f$ such that
$\mathbb{K}(f_1,\ldots,f_m)=\mathbb{K}(f)$, and second to
decompose the rational function $f$, i.e., to find $g,h\in
\mathbb{K}(y)$ such that,  $\mathbb{F}= \mathbb{K}(h)$ with
$f=g(h)$. Constructive proofs of L\"{u}roth's theorem can be found
in \cite{Net1895}, \cite{sed86} and \cite{AGR95}. Algorithms for
decomposition of univariate rational functions can be found in
\cite{Zip91} and \cite{AGR95}. \par

In the multivariate case, the problem is: given $f_1,\ldots,f_m$
in $\mathbb{K}(\mathbf{x})$ we wish to know if there exists a
proper intermediate field $\mathbb{F}$ such that
$\mathbb{K}(f_1,\ldots,f_m)\subset \mathbb{F}\subset
\mathbb{K}(\mathbf{x})$ with
$\mathrm{tr.deg}(\mathbb{F}/\mathbb{K})=1$; and in the affirmative
case, to compute it. A central result is the following
generalization of L\"uroth's theorem:

\begin{theorem}[Extended L\"uroth's Theorem]\label{extended}
Let $\mathbb{F}$ be a field such that $\mathbb{K}\subset
\mathbb{F}\subset \mathbb{K}(\mathbf{x})$. If
$\mathrm{tr.deg}(\mathbb{F}/\mathbb{K})=1$, then there exists
$f\in \mathbb{K}(\mathbf{x})$ such that
$\mathbb{F}=\mathbb{K}(f)$.
\end{theorem}

Such an $f$ is called a L\"uroth's generator of the field
$\mathbb{F}$. This theorem was first proved in \cite{Gor1887} for
characteristic zero and in \cite{Igu95} in general, see also
\cite{Sch82} Theorem 3. Using Gr\"obner basis computation, the
paper \cite{Ste00} provides an algorithm to compute a L\"uroth's
generator, if it exits. See also \cite{Sarito00} for  another
algorithmic proof of this result.
 In this paper we present a new algorithm, which only
requires to compute  $\mathrm{gcd}$'s,  to detect if a unirational
field has transcendence degree 1 and, in the affirmative case, to
compute a L\"uroth's generator. We also present a constructive
proof of the above theorem for polynomials (see \cite{Noe15}): if
the unirational field contains a non-constant polynomial, then  it
is generated by a polynomial. \par

By the Extended L\"uroth's theorem, to find an intermediate field
of transcendence degree one is equivalent to the following: first
to find a L\"{u}roth's generator $f$, i.e.,
$\mathbb{K}(f_1,\ldots,f_m)=\mathbb{K}(f)$, if it exists, and
second to decompose the multivariate rational function $f$, i.e.,
to find $g\in \mathbb{K}(y)$ and $h\in \mathbb{K}(\mathbf{x})$
such that $f=g(h)$ in a nontrivial way. The pair $(g,h)$ is called
a uni--multivariate decomposition of $f$.  We present two
algorithms to compute a nontrivial uni--multivariate decomposition
of a multivariate rational function, if it exits.
\par

This paper is divided in four sections. In section 2 we state the
proof of the Extended L\"uroth's theorem and its polynomial
version. In section 3 we present and analyze two algorithms to
compute a uni--multivariate decomposition of a rational function,
if it exists. In section 4 we discuss the performance of these
algorithms.

\section{The Extended L\"uroth Theorem}
In this section we present an algorithm to the following
computational problem: given $f_1,\ldots,f_m\in
\mathbb{K}(\mathbf{x})$, to compute a L\"uroth generator $f$ for
$\mathbb{F}=\mathbb{K}(f_1,\ldots,f_m)$ if it exists, moreover we
detect if $\mathbb{F}$ contains a non-constant polynomial, and in
the affirmative case we find a generating polynomial. We start
with the following definition:

\begin{definition} Let
$p\in\mathbb{K}[x_1,\ldots,x_n,y_1,\ldots,y_n]
=\mathbb{K}[\mathbf{x},\mathbf{y}]$ be a non--constant polynomial.
We say that $p$ is {\bf near--separated} if there exist
non--constant polynomials $r_1,s_1\in
\mathbf{K}[\mathbf{x}]=\mathbb{K}[x_1,\ldots,x_n]$ and $r_2,s_2\in
\mathbf{K}[\mathbf{y}]= \mathbb{K}[y_1,\ldots,y_n]$, such that
neither $r_1,s_1$ are associated, nor $r_2,s_2$ are associated and
$p=r_1s_2-r_2s_1$. In the particular case when
$p=r(x_1,\ldots,x_n)s(y_1,\ldots,y_n)-s(x_1,\ldots,x_n)r(y_1,\ldots,y_n)$,
we say that $p$ is a {\bf symmetric near--separated} polynomial.
We say that $(r,s)$ is a {\bf symmetric near--separated
representation} of $p$.
\end{definition}

In this paper, $\mathrm{deg}_{x_1,\ldots,x_n}$ will denote the
total degree with respect to the variables $x_1,\dots,x_n$ and
$\mathrm{deg}$ will denote the total degree with respect to all
the variables. Also, given a rational function $f$ we will also de
note as $f_n,f_d$ the numerator and denominator of $f$,
respectively.

In the following theorem we give some basic properties of
near--separated polynomials, for later use.

\begin{theorem}\label{prop-casisep}
Let $p\in\mathbb{K}[\mathbf{x},\mathbf{y}]$ be a near--separated
polynomial and $r_1,s_1,r_2,s_2$ as in the above definition. Then
\begin{description}
\item[(i)] If $\mathrm{gcd}(r_1,s_1)=1$ and
$\mathrm{gcd}(r_2,s_2)=1$, $p$ has no factors in
$\mathbb{K}[\mathbf{x}]$ or $\mathbb{K}[\mathbf{y}]$.
\item[(ii)]
$\mathrm{deg}_{x_1,\ldots,x_n}\,p=\mathrm{max}\{\mathrm{deg}\
r_1,\mathrm{deg}\ s_1\}$ and
$\mathrm{deg}_{y_1,\ldots,y_n}\,p=max\{\mathrm{deg}\
r_2,\mathrm{deg}\ s_2\}$.
\item[(iii)] If $p$ is symmetric and $(\alpha_1,\ldots,\alpha_n)\in \mathbb{K}^n$
verifies\linebreak$p(x_1,\ldots,x_n,\alpha_1,\ldots,\alpha_n)\neq
0$, then there exists $(r,s)$, a symmetric near--separated
representation of $p$, such that
\[r(\alpha_1,\ldots,\alpha_n)=0\quad {\rm and }\quad s(\alpha_1,\ldots,\alpha_n)=1. \]
\item[(iv)] If $p$ is symmetric, the coefficient of $x_k^{i_0}y_k^{j_0}$ in $p$
is the near--separated polynomial
\[
a_{i_0}\,b_{j_0} - \\ b_{i_0}\,a_{j_0},
\]
\noindent where $a_i$ is the coefficient of $x_k^i$ in $r$ and
$b_i$ is the coefficient of $x_k^i$ in $s$.
\end{description}
\end{theorem}
\begin{proof}
{\bf (i)} Suppose $v\in\mathbb{K}[x_1,\ldots,x_n]$ is a
non--constant factor of $p$. Then there exists $i$ such that
$\mathrm{deg}_{x_i}v\geq~1$. Without loss of generality we will
suppose that $i=1$. Let $\alpha$ be a root of $v$, considering $p$
as a univariate polynomial in the variable $x_1$, in a suitable
extension of $\mathbb{K}[x_2,\ldots,x_n]$. If $\alpha $ is a root
of any of the polynomials $r_1$ or $s_1$, then it is also a root
of the other. This is a contradiction, because
$\mathrm{gcd}(r_1,s_1)=1$. Therefore $\alpha$ is neither a root of
$r_1$ nor of $s_1$. Then,
\[\displaystyle{\frac{r_1(\alpha,x_2,\ldots,x_n)}{s_1(\alpha,x_2,\ldots,x_n)}} =
\displaystyle{\frac{r_2(y_1,\ldots,y_n)}{s_2(y_1,\ldots,y_n)}}\in\mathbb{K}.\]
A contradiction again, since $r_2,s_2$  are not associated in
$\mathbb{K}$.

{\bf (ii)} If $\mathrm{deg}\ r_1\neq\mathrm{deg}\ s_1$, the
equality is trivial. Otherwise, if $\mathrm{deg}\
r_1=\mathrm{deg}\ s_1>\mathrm{deg}_{x_1,\ldots,x_n}p$, the terms
with greatest degree with respect to $x_1,\ldots,x_n$ vanish. This
is a contradiction, because $r_2,s_2$ are not associated. The
proof is the same for $r_2,s_2$.

{\bf (iii)} Let $(r,s)$ be a representation of $p$.
\begin{description}
\item [--] If
$r(\alpha_1,\ldots,\alpha_n)=0$, since
$p(x_1,\ldots,x_n,\alpha_1,\ldots,\alpha_n)\neq 0$, we have
$s(\alpha_1,\ldots,\alpha_n)\neq 0$. Then we have a new
near--separated representation:
\[\left(r\,s(\alpha_1,\ldots,\alpha_n),\displaystyle{\frac{s}{s(\alpha_1,\ldots,\alpha_n)}}\right).\]

\item [--] If $s(\alpha_1,\ldots,\alpha_n)=0$, then the representation $(-s,r)$ we are in the previous case.

\item [--] If $r(\alpha_1,\ldots,\alpha_n),s(\alpha_1,\ldots,\alpha_n)\neq 0$, then we
consider the representation
\[\left(r\,s(\alpha_1,\ldots,\alpha_n)-s\,r(\alpha_1,\ldots,\alpha_n),
\displaystyle{\frac{s}{s(\alpha_1,\ldots,\alpha_n)}}\right).\]
\end{description}

{\bf (iv)} This is a simple routine confirmation.
\end{proof}

Now, we state an important theorem that relates
uni--multi\-variate decompositions to near--separated polynomials,
that is proved in \cite{Schi95}:

\begin{theorem}\label{schicho}
Let $\mathbb{A}=\mathbb{K}(\mathbf{x})$ and
$\mathbb{B}=\mathbb{K}(\mathbf{y})$ be rational function fields
over $\mathbb{K}$. Let $f,h\in \mathbb{A}$ and $f',h'\in
\mathbb{B}$ be non--constant rational functions. Then the
following statements are equivalent:
\begin{description}
\item[A)] There exists a rational function $g\in \mathbb{K}(t)$ satisfying
$f=g(h)$ and $f'=g(h')$.
\item[B)] $h-h'$ divides $f-f'$ in $\mathbb{A}\otimes_\mathbb{K} \mathbb{B}$.
\end{description}
\end{theorem}

As a consequence, a rational function $f\in\mathbb{K}(\mathbf{x})$
verifies $f=g(h)$ for some $g,h$ if and only if $
h_n(\mathbf{x})\:h_d(\mathbf{y})-h_d(\mathbf{x})\:h_n(\mathbf{y})$
divides $
f_n(\mathbf{x})\:f_d(\mathbf{y})-f_d(\mathbf{x})\:f_n(\mathbf{y})$.

Given an admissible monomial ordering $>$ in a polynomial ring and
a nonzero polynomial $G$ in that ring, we denote by $\mathrm{lm}\
G$ the leading monomial of $G$ with respect to $>$ and
$\mathrm{lc}\ G$ its leading coefficient.

\begin{algorithm}\label{alg-netto}
\
\begin{description}
\item[Input:] $f_1,\ldots,f_m\in \mathbb{K}(\mathbf{x})$.
\item[Output:] $f\in \mathbb{K}(\mathbf{x})$ such that
$\mathbb{K}(f)=\mathbb{F}=\mathbb{K}(f_1,\ldots,f_m)$, if it
exists. Otherwise, returns null.
\end{description}
\begin{description}
\item[A] Let $>$ be a graded lexicographical ordering for\linebreak
$\mathbf{y}=(y_1,\ldots,y_n)$. Let $i=m$.
\item[B] Let
$F_k={f_k}_n(\mathbf{y})-f_k(\mathbf{x}){f_k}_d(\mathbf{y})$ for
$k=1,\ldots,i$.
\item[C] Compute $H_i=\mathrm{gcd}(\{F_k,\ k=1,\ldots,i\})$ with $\mathrm{lc}\ H_i=1$.
\item[D]
\begin{description}
\item [--] If $H_i=1$, \textbf{RETURN NULL} ($\mathbb{F}$ does not have transcendence degree $1$ over $\mathbb{K}$).
\item [--] If there exists $j\in\{1,\ldots,i\}$ such that $\mathrm{lm}\ H_i=\mathrm{lm}\ F_j$, then \textbf{RETURN} $f_j$.
\item [--] Otherwise, let $f_{i+1}$ be a coefficient of $H_i$ in
$\mathbb{F}\setminus\mathbb{K}$. Increase $i$ and go to
\textbf{B}.
\end{description}
\end{description}
\end{algorithm}
{\bf Correctness proof.} If $\mathbb{F}$ has transcendence degree
$1$ over $\mathbb{K}$, we can write $\mathbb{F}=\mathbb{K}(f)$. By
Theorem \ref{schicho},
\[f_n(\mathbf{y})-f(\mathbf{x})f_d(\mathbf{y})\] divides
$H_i$ for any $i$. Therefore, $H_i$ is non--constant if a L\"uroth
generator exists.

If there are $i,j$ such that $\mathrm{lm}\ H_i=\mathrm{lm}\ F_j$,
then $F_j$ is a greatest common divisor of $\{F_k,
k=1,\ldots,i\}$. Therefore, $F_j$ divides $F_k$ for every $k$. Fix
such a $k$. Let $q=\overline{{f_k}_n(\mathbf{y})}^{\{F_j\}},
s=\overline{{f_k}_d(\mathbf{y})}^{\{F_j\}}$ the normal form with
respect to the monomial ordering $>$; then there exist $p,q,r,s\in
\mathbb{F}[\mathbf{y}]$ such that
\[\begin{array}{l}
{f_j}_n(\mathbf{y})=p(\mathbf{y})F_i-q(\mathbf{y}) \\
{f_j}_d(\mathbf{y})=r(\mathbf{y})F_i-s(\mathbf{y})
\end{array}\]
where $\mathrm{lm}\ F_j$ does not divide any monomial of $q$ or
$s$. By theorem \ref{prop-casisep}(i), $q,s\neq 0$. By the
definition in step B,
\[F_k=F_j(p-r\,f_k(\mathbf{x}))-(q-s\,f_k(\mathbf{x})).\]
Hence $F_j$ divides $q-s\,f_k(x_1,\ldots,x_n)$ and we conclude
that $q-s\,f_k(\mathbf{x})=0$, since otherwise we would get that
$\mathrm{lm}\ F_j$ divides $\mathrm{lm}\
(q-s\,f_k(x_1,\ldots,x_n))$, which contradicts the choice of the
polynomials $q,s$. Thus
$f_k(x_1,\ldots,x_n)=\displaystyle{\frac{q}{s}}\in
\mathbb{F}=\mathbb{K}(f_j)$.

Now we suppose that $\mathrm{lm}\ H_i<\ \mathrm{lm}\ F_k$ for all
$k$. Again, fix a value for $k$. Then there exists a $C\in
\mathbb{F}[\mathbf{y}]\setminus\mathbb{F}$ such that $F_k=H_iC$.
Let $d,\alpha$ be the lowest common multiples of the denominators
of the coefficients of $H_i$ and $C$, respectively. Then
$D=H_id,C'=\alpha C\in \mathbb{K}[\mathbf{x},\mathbf{y}]$. Since
$H_i$ is monic, the polynomial $D$ is primitive. Then,
\[{f_k}_n(\mathbf{y})\:{f_k}_d(\mathbf{x})-{f_k}_n(\mathbf{x})\:{f_k}_d
(\mathbf{y})=\displaystyle{\frac{D}{d}}\:\displaystyle{\frac{C'}{\alpha}}\:{f_k}_d\]
and by theorem \ref{prop-casisep},
\[{f_k}_n(\mathbf{y})\:{f_k}_d(\mathbf{x})-{f_k}_n(\mathbf{x})\:{f_k}_d
(\mathbf{y})=D\widehat{C},\] $\widehat{C}\in
\mathbb{K}[\mathbf{x},\mathbf{y}]$. On one hand, $D\not\in
\mathbb{K}[\mathbf{y}]$, thus $D$ (and $H_i$) have a non--constant
coefficient. On the other hand, $\widehat{C}\not\in
\mathbb{K}[\mathbf{y}]$, then the non--constant coefficients of
$D$ in the ring $\mathbb{K}(\mathbf{x})[\mathbf{y}]$ have smaller
degree than that of $f_k(\mathbf{x})$. The choice of $d$ assures
that the coefficients of $H$ have smaller degree than $f_k$.
Summarizing, there exists a coefficient $a\in \mathbb{F}$ of $H_i$
that can be added to the list of generators and has smaller degree
than them. If tr.deg $(\mathbb{F}/\mathbb{K})=1$, $H_i$ is
non--constant for all $i$, and the generator has smaller degree
than the others. Therefore, the algorithm ends in a finite number
of steps. \par

Finally, we note that complexity is dominated in the step {\bf C}
by  computing $\mathrm{gcd}$'s of multivariate polynomials, so the
algorithm is polynomial in the degree of the rational functions
and in $n$ (see \cite{GG99}).

From the fact that the L\"uroth generator can be found with only
some $\mathrm{gcd}$ computations, we obtain that if $f$ is a
L\"uroth generator of $\mathbb{K}(f_1,\ldots,f_n)$ then it is also
a L\"uroth generator of $\mathbb{K'}(f_1,\ldots,f_n)$ for any
field extension $\mathbb{K'}$ of $\mathbb{K}$,
$\mathbb{K}\subset\mathbb{K'}$.

\
\begin{example}
Let $\mathbb{Q}(f_1,f_2)\subset \mathbb{Q}(x,y,z)$ where
\[\begin{array}{l}
f_1=\displaystyle{\frac{y^2x^4-2y^2x^2z+y^2z^2+x^2-2xz+z^2}{yx^3-yxz-yzx^2+z^2y}}
\\
\\
f_2=\displaystyle{\frac{y^2x^4-2y^2x^2z+y^2z^2}{x^2-2xz+yx^3-yxz+z^2-yzx^2+z^2y}}.
\end{array}
\]
Let
\[
F_i={f_i}_n(s,t,u)-f_i(x,y,z){f_i}_d(s,t,u)\,,\ \ i=1,2.
\]
Compute
\[
H_2=\mathrm{gcd}(F_1,F_2)=-tu+s^2t+\displaystyle{\frac{x^2y-zy}{x-z}
}u+\displaystyle{\frac{-x^2y+zy}{x-z}}s.
\]
Since $\mathrm{lm}\ H_2<\mathrm{lm}\ F_i$ with respect to the
lexicographical ordering $s>t>u$, we take a non--constant
coefficient of $H_2$: $f_3=\displaystyle{\frac{x^2y-zy}{x-z}}$.
Now
\[H_3=-tu+s^2t+
\displaystyle{\frac{x^2y-zy}{x-z}}u+\displaystyle{\frac{-x^2y+zy}{x-z}}s
\]
and $H_3=F_3$, since $H_3=H_2$. The algorithm returns $f_3$, a
L\"uroth generator of $\mathbb{Q}(f_1,f_2)$.
\end{example}

It is important to highlight that when the field $\mathbb{F}$
contains a non--constant  polynomial you can  compute a polynomial
as a generator, and this generator neither depends on the ground
field $\mathbb{K}$. This result  was proved in \cite{Noe15}, for
zero characteristic. A general proof can be
 found in \cite{Sch82}.

\begin{algorithm}\label{alg-poly}
\
\begin{description}
\item[Input:] $f_1,\ldots,f_m\in \mathbb{K}(\mathbf{x})$.
\item[Output:] $f\in \mathbb{K}[\mathbf{x}]$ such that
$\mathbb{K}(f)=\mathbb{F}=\mathbb{K}(f_1,\ldots,f_m)$, if it
exists. Otherwise, returns null.
\end{description}
\begin{description}
\item[A] Compute a L\"uroth generator $f$ of $\mathbb{K}(f_1,\dots,f_m)$
using Algorithm~\ref{alg-netto}.
\item[B]  Let  $s$ be the degree of $f'$.
\begin{description}
\item [---] If $s>\mathrm{deg}\ f_n ' $ and $f_n ' $  is not constant, return null. Otherwise, let $f=1/f ' $.
\item [---] If $s>\mathrm{deg}\ f_d ' $ and $f_d ' $  is not constant, return null. Otherwise, let $f=f ' $.
\item [---] Let ${f_n^{(s)}} ', {f_d^{(s)}} ' $ be the homogeneous
components of degree $s$ of $f_n',f_d ' $, respectively. Let
$a=\displaystyle{\frac {{f_n^{(s)}} '} {{f_d^{(s)}} '}} $. If $a $
or $f_n'-af_d ' $  are not constant, return null. \par
 Otherwise,
let $f=\displaystyle{\frac{1}{y-a}} \circ f' $.
\end{description}
\end{description}
\end{algorithm}
{\bf Correctness proof.} Once a L\"uroth's generator   has been
computed, take a generator $f $ with  degree $m $ such that if
$f=\displaystyle { \frac {f_n} {f_d}} $ and

\[
\begin{array}{rl}
f_n&=f_n^{(s) }+\cdots+f_n^{(0)}, \\ f_d&=f_d^{(s)
}+\cdots+f_d^{(0)},
\end{array}
\]
\noindent the sum in homogeneous polynomials, then either
$f_d^{(s)} =0 $  or $\displaystyle { \frac{f_n^{(s)}} {f_d^{(s)}}
} \not\in \mathbb{K}$.

\noindent If $p \in \mathbb{K}(f)$ is a polynomial, then there
exists $g\in \mathbb{K}(y) $ with degree $r $ such that $p=g(f) $.
If $g=\displaystyle{\frac {a_ry^r+\cdots+a_0}{b_ry^r+\cdots+b_0}}
$,
\[
\begin{array}{rl}
p&=\displaystyle{\frac {a_rf_n^r+\cdots+a_0f_d^r}
{b_rf_n^r+\cdots+b_0f_d^r} } \\ & = \displaystyle{\frac
{a_r{(f_n^{(s) }+\cdots+f_n^{(0)})} ^r+\cdots+a_0(f_d^{(s)
}+\cdots+f_d^{(0)}) ^r} {b_r(f_n^{(s) }+\cdots+f_n^{(0)})
^r+\cdots+b_0(f_d^{(s) }+\cdots+f_d^{(0)}) ^r} }.
\end{array} \]

\noindent Since $p $  is a polynomial, the degree of the previous
denominator
  is smaller than the degree of the numerator. Therefore
$b_r{f_n^{(s)}} ^r+\cdots+b_0{f_d^{(s)}} ^r=0 $.

\noindent If $f_d^{(s)} =0 $ then $b_r=0 $ and
$p=\displaystyle{\frac {a_rf_n^r+\cdots+a_0f_d^r}
{f_d(b_{r-1}f_n^{r-1}+\cdots+b_0f_d^{r-1})}}. $  Hence $f_d $
divides the numerator of p, and therefore divides $f_n $. This
proves that $f $  is a polynomial.

\noindent If, on the contrary, $f_d^{(s)} \not=0 $,
$g_d\left(\displaystyle { \frac{f_n^{(s)}} {f_d^{(s)}}} \right)=0
$. Contradiction, since $\displaystyle { \frac{f_n^{(s)}}
{f_d^{(s)} }} \not\in \mathbb{K} $.

\section{Two uni--multivariate decomposition algorithms}

 We define the \textbf{degree} of a rational
function $f=f_n/f_d \in \mathbb{K}(\mathbf{x})$ as $\mathrm{deg}\
f=\mathrm{max}\ \{\mathrm{deg}\ f_n,\mathrm{deg}\ f_d\}$ if
$\mathrm{gcd}(f_n,f_d)=1$.  The following definition was
introduced in \cite{Gat90} for polynomials.
\begin{definition} Let
$f,h \in \mathbb{K}(\mathbf{x})$ and $g \in \mathbb{K}(y)$ such
that $f=g(h)$. Then we say that $(g,h)$ is a uni-multivariate
decomposition of $f$. It is non-trivial if $ 1 <\mathrm{deg}\  h <
\mathrm{deg}\ f$. The rational function is uni-multivariate
decomposable if there exits a non-trivial decomposition.
\end{definition}

If $f$ is a polynomial having a nontrivial uni-multivariate
decomposition, then  by Algorithm \ref{alg-poly} we get that there
exits a uni-multivariate decomposition $(g,h)$ with $g$ and $h$
polynomials. The paper \cite{Gat90} provides an algorithm to
compute a nontrivial uni-multivariate decomposition of a
polynomial $f$ of degree $m$ that only requires $O(nm(m+1)^n\log
\; m) $ arithmetic operations in the ground field $\mathbb{K}$.

The known techniques for decomposition all divide the problem into
two parts. Given $f$, in order to find a decomposition $f=g(h)$,
\begin{enumerate}
\item
  \label{1.}
  one first computes candidates $h$,
\item
  \label{2.}
  then computes $g$ given $h$.
\end{enumerate}

Determining $g$ from $f$ and $h$ is a subfield membership problem.
The paper \cite{Swe93} gives a solution to this part. We also
present another faster method, that only requires solving a linear
system of equations. Usually, the harder step is to find
candidates for $h$. One goal in decomposition is to have
components of smaller degree than the composed polynomial. This
will be the case here.

\subsection{Preliminary results}

First, we state some results that will be used in the algorithms
presented later. On the properties to highlight out of
uni-multivariate decomposition is the good behaviour of the degree
with respect to this composition.

\begin{theorem}\label{muldeg}
Let $g\in \mathbb{K}(y)$ and $h\in \mathbb{K}(\mathbf{x})$, and
$f=g(h)$. Then $\mathrm{deg}\ f=\mathrm{deg}\ g\cdot\mathrm{deg}\
h$.
\end{theorem}
\begin{proof} Let $g=\displaystyle{\frac{g_n}{g_d}}$ and
$h=\displaystyle{\frac{h_n}{h_d}}$ with $\mathrm{gcd}(g_n,g_d)=1$
and $\mathrm{gcd}(h_n,h_d)=1$. Then there exist polynomials
$A,B\in \mathbb{K}[y]$ such that
\[g_n(y)\,A(y)+g_d(y)\,B(y)=1.\] Homogenizing the polynomials
$g_n,g_d,A,B$ we obtain, respectively, the bivariate polynomials
$\widetilde{g}_n(y_1,y_2),\widetilde{g}_d(y_1,y_2)$,
$\widetilde{A}(y_1,y_2),\widetilde{B}(y_1,y_2)$ verifying
\[\widetilde{g}_n(y_1,y_2)\:\widetilde{A}(y_1,y_2)\:y_2^u+\widetilde{g}_d(y_1,y_2)\:\widetilde{B}(y_1,y_2)\:y_2^v=y_2^w\]
with either $u=0$ or $v=0$ and $w=$max$\{u,v\}$.
Therefore,\[\widetilde{g}_n(h_n,h_d)\:\widetilde{A}(h_n,h_d)\:h_d^u+\widetilde{g}_d(h_n,h_d)\:\widetilde{B}(h_n,h_d)\:h_d^v=h_d^w.\]
If $d$ is an irreducible factor of
$\mathrm{gcd}(\widetilde{g}_n(h_n,h_d),\widetilde{g}_d(h_n,h_d))$,
then $d$ divides $h_d$. On the other hand, $d$ divides
$\widetilde{g}_n(h_n,h_d)$ and $\widetilde{g}_d(h_n,h_d)$; this
implies that $d$ divides $h_n$. As a consequence,
$\mathrm{gcd}(\widetilde{g}_n(h_n,h_d),\widetilde{g}_d(h_n,h_d))=1.$
So,
\[f=\displaystyle{\frac{\widetilde{g}_n(h_n,h_d)}{\widetilde{g}_d(h_n,h_d)}}\:h_d^a\,,\
|a|=|\mathrm{deg}\ h_n-\mathrm{deg}\ h_d|\] is in reduced form.
Without loss of generality, we can take deg $g_n=r_n\geq r_d=$deg
$g_d$ with
\[
\begin{array}{l}
g_n(y)=a_{r_n}\,y^{r_n}+\cdots+a_0 \\
g_d(y)=b_{r_d}\,y^{r_d}+\cdots+b_0.
\end{array}
\]
Then, $\mathrm{deg}\ f=\mathrm{max}\ \{\mathrm{deg}\
\widetilde{g}_n(h_n,h_d),\mathrm{deg}\
\widetilde{g}_d(h_n,h_d)\,h_d^{r_n-r_d}\}$ and
\[\begin{array}{l}
\widetilde{g}_n(h_n,h_d)=a_{r_n}h_n^{r_n}+\cdots+a_0h_d^{r_n} \\
\widetilde{g}_d(h_n,h_d)=b_{r_d}h_n^{r_d}+\cdots+b_0h_d^{r_d}.
\end{array}
\]
If $\mathrm{deg}\ \widetilde{g}_n(h_n,h_d)=r_n\,\mathrm{deg}\ h$,
we immediately obtain that $\mathrm{deg}\ f=\mathrm{deg}\
g\,\mathrm{deg}\ h$.
If $\mathrm{deg}\ \widetilde{g}_n(h_n,h_d)<r_n\mathrm{deg}\ h$,
then $s=\mathrm{deg}\ h_n=\mathrm{deg}\ h_d$. Write
\[\begin{array}{l}
h_n=h_n^{(s)}+h_n^{(s-1)}+\cdots+h_n^{(0)} \\
h_d=h_d^{(s)}+h_d^{(s-1)}+\cdots+h_d^{(0)}
\end{array}
\]
where $h_n^{(j)},h_d^{(j)}$ are the homogeneous components of
$h_n$ and $h_d$ with degree $j$, respectively. Since the degree
decreases, $\mathrm{deg}\ \widetilde{g}_n(h_n^{(s)},h_d^{(s)})=0$
and $\mathrm{deg}\
\widetilde{g}_n\left(\displaystyle{\frac{h_n^{(s)}}{h_d^{(s)}}}\right)=0$.
Therefore, $\displaystyle{\frac{h_n^{(s)}}{h_d^{(s)}}}\in
\mathbb{K}$. In this case, you can take $h'\in
\mathbb{K}(\mathbf{x})$ a rational function with $\mathrm{deg}\
h=\mathrm{deg} h',\ \mathrm{deg}\ h_n'\neq\mathrm{deg}\ h'_d$ and
such that $f=g'(h')$ for some $g'\in\mathbb{K}(y)$ with
$\mathrm{deg}\ g=\mathrm{deg}\ g'$. Under these hypothesis, we
proved before that $\mathrm{deg}\ f=\mathrm{deg}\
g'\,\mathrm{deg}\ h'=\mathrm{deg}\ g\,\mathrm{deg}\ h$.
\end{proof}
\begin{corollary}\label{primeform}
Let $g=g_n/g_d$ with $g_n=a_uy^u+\cdots+a_0,g_d=b_vy^v+\cdots+b_0$
and $h=h_n/h_d$ verifying
$\mathrm{gcd}(g_n,g_d)=\mathrm{gcd}(h_n,h_d)=1$. If
$f=f_n/f_d=g(h)$ with
\[
\begin{array}{l}
f_n=(a_u\,h_n^u+\cdots+a_0\,h_d^u)\,h_d^{\mathrm{max}\{v-u,0\}} \\
f_d=(b_v\,h_n^v+\cdots+b_0\,h_d^v)\,h_d^{\mathrm{max}\{u-v,0\}}
\end{array}
\]
then $\mathrm{gcd}(f_n,f_d)=1$.
\end{corollary}
\begin{proof}
It is easy to prove that \[\mathrm{deg}\ f_n,\mathrm{deg}\
f_d\leq\mathrm{max}\{u,v\}\cdot\mathrm{max}\{\mathrm{deg}\
h_n,\mathrm{deg}\ h_d\}.\] If $\mathrm{gcd}(f_n,f_d)\neq 1$, then
$\mathrm{deg}\ f<\mathrm{deg}\ g\,\mathrm{deg}\ h$, contradicting
theorem \ref{muldeg}.
\end{proof}
\begin{corollary}\label{compute-g}
Given $f$ and $h$, if there exists $g$ such that $f=g(h)$, then
$g$ is unique. Furthermore, it can be computed from $f$ and $h$ by
solving a linear system of equations.
\end{corollary}
\begin{proof}
If $f=g_1(h)=g_2(h)$, then $(g_1-g_2)(h)=0$, and by theorem
\ref{muldeg}, $\mathrm{deg}\ (g_1-g_2)=0$, thus $g_1-g_2$ is
constant. It is then clear that it must be $0$, that is,
$g_1=g_2$. Again by theorem \ref{muldeg}, the degree of $g$ is
determined by those of $f$ and $h$. We can write $g$ as a function
with the corresponding degree and undetermined coefficients.
Equating to zero the coefficients of the numerator of $f-g(h)$, we
obtain a linear homogeneous system of equations in the
coefficients of $g$. If we compute a non--trivial solution to this
system, we find $g$.
\end{proof}

\begin{definition} Let  $f\in  \mathbb{K}(\mathbf{x})$ be a rational function. Two
uni--multivariate decompositions $(g,h)$ and $(g',h ')$ of $f $
are {\bf equivalent} if there exists a composition unit $l\in
\mathbb{K}(y)$, i.e., $\mathrm{deg}\ l =1$,  such that $h=l(h') $.
\end{definition}
\begin{corollary} \label{cuer-umfr} Let  $f\in \mathbb{K}(\mathbf{x}) $
 be a non--constant rational function. Then
the equivalence classes of uni--multivariate decompositions of $f$
correspond bijectively to intermediate fields $\mathbb{F}$,
$\mathbb{K}(f) \subset \mathbb{F} \subset \mathbb{K}(\mathbf{x})$,
with transcendence degree $1$ over $\mathbb{K} $.
\end{corollary}

\begin{proof} The bijection is
\[
\begin{array}{ccc}
\{[(g,h)], \;f=g(h) \} & \longrightarrow & \{\mathbb{K}(f)\subset
\mathbb{F},{\rm }(\mathbb{F} / \mathbb{K})=1 \}. \\
\left[(g,h)\right]&\longmapsto & \mathbb{F}=\mathbb{K}(h)
\end{array}
\]
Suppose we have a uni--multivariate decomposition $(g,h) $ of $f
$. Since  $f=g(h) $, $\mathbb{F}=\mathbb{K}(h) $  is an
intermediate field of $\mathbb{K}(f)\subset \mathbb{K}(\mathbf{x})
$ with transcendence  degree $1 $ over
 $\mathbb{K} $. Also, if $(g',h ') $  is equivalent to
$(g,h) $,
 $h=l(h') $ for some composition unit $l\in \mathbb{K}(y)$. Consequently, $h'=l^{-1}(h) $ and $\mathbb{K}(h)
= \mathbb{K}(h') $.  Let $(g,h) $ and $(g',h ') $ be two
uni--multivariate decompositions   of $f $ such that
$\mathbb{K}(h)=\mathbb{K}(h ') $. Then there exists rational
functions $l,l ' \in \mathbb{K}(y) $ such that $h=l(h') $ and
$h'=l'(h) $.
 By   theorem~\ref{muldeg}, $\mathrm{deg}\  l(l')=1 $ and
$\mathrm{deg}\ l =\mathrm{deg}\ l'=1 $. By the uniqueness (see
Corollary \ref{compute-g}) of the left component, $y=l(l') $. So,
$l\in \mathbb{K}(y) $  is a composition unit and $(g,h) $, $(g',h
') $ are equivalent.
 By Theorem~\ref{extended},  there exist $h\in \mathbb{K}(\mathbf{x}) $ and
$g\in \mathbb{K}(y)$ such that $\mathbb{F}=\mathbb{K}(h) $ and
$f=g(h) $.
\end{proof}

\subsection{First algorithm}
We now proceed with the first algorithm for computing candidates
$h=h_n/h_d$. This algorithm is based on Theorem \ref{schicho}.
Since $
h_n(\mathbf{x})\:h_d(\mathbf{y})-h_d(\mathbf{x})\:h_n(\mathbf{y})$
divides $
f_n(\mathbf{x})\:f_d(\mathbf{y})-f_d(\mathbf{x})\:f_n(\mathbf{y})$,
one can compute candidates for $h$ from $f$ merely looking at the
near-separated divisors
$H=r(\mathbf{x})\:s(\mathbf{y})-r(\mathbf{y})\:s(\mathbf{x}) $.
Next, the problem is: given a multivariate polynomial
$H=(\mathbf{x}, \mathbf{y}) $, how can one determine if it is a
symmetric near-separated polynomial? This is a consequence of
theorem \ref{prop-casisep}:

\begin{corollary}
Given a polynomial $p\in\mathbb{K}[\mathbf{x}\,,\,\mathbf{y}]$, it
is possible to find a near--separated representation $(r,s)\in
\mathbb{K}[\mathbf{x}]^2$ of $p$, if it exists, by solving a
linear system of equations with coefficients in $\mathbb{K}$.
Moreover, any other solution $(r',s')$ of this linear system of
equations gives an equivalent decomposition.
\end{corollary}
\begin{algorithm}\label{alg1}
\
\begin{description}
\item [Input:] $f\in \mathbb{K}(\mathbf{x})$.
\item [Output:] A uni--multivariate decomposition $(g,h)$ of $f$, if it
exists.
\end{description}
\begin{description}
\item[A] Factor the symmetric polynomial
\[p=f_n(\mathbf{x})\:f_d(\mathbf{y})-f_d(\mathbf{x})\:f_n(\mathbf{y}).\]
Let $D=\{H_1,\ldots,H_m\}$ the set of factors of $p$ (up to
product by constants). Let $i=1$.
\item[B] Check if $H_i$ is a symmetric near--separated polynomial.
If $H=r(\mathbf{x})\:s(\mathbf{y})-r(\mathbf{y})\:s(\mathbf{x})$,
then $h=\displaystyle{\frac{r}{s}}$ is a right--component for $f$;
compute the left component $g$ by solving a linear system (see
Corollary \ref{compute-g}) and \textbf{RETURN} $(g,h)$.
\item[C] If $i<m$, then increase $i$ and go to \textbf{B}.
Otherwise, \textbf{RETURN NULL} ($f$ is uni--multivariate
indecomposable).
\end{description}
\end{algorithm}
\begin{example}\label{alg1-ex1}
Let
\[
\begin{array}{rcl}
f&=&4z^4y^2-8z^3y^3+8z^2yx+4z^2y^4-8zy^2x\\
&&+4x^2-2z^2y+2zy^2-2x+10.
\end{array}
\]
The factorization of the polynomial $f(x,y,z)-f(s,t,u)$ is \[
\begin{array}{l}2\,
\left (2x-1+2s-2ut^2+2u^2t-2zy^2+2z^2y \right)\\ \left
(x-s+z^2y-zy^2-u^2t+ut^2\right).
\end{array}
\]
The first factor $f_1=2x-1+2s-2ut^2+2u^2t-2zy^2+2z^2y$ is not
symmetric near--separated because $f_1(x,y,z,x,y,z)\neq 0$. On the
other hand, the second factor $f_2=x-s+z^2y-zy^2-u^2t+ut^2$ does
satisfy $f_2(x,y,z,x,y,z)=0$. Note that by a previous remark, the
components of the decomposition can be considered as polynomials.
Then $f_2$ can be written as $f_2=h(x,y,z)-h(s,t,u)$. Taking
$h(x,y,z)=f_2(x,y,z,0,0,0)=x+z^2y-zy^2$, we check that it
satisfies the previous equation (see theorem~\ref{prop-casisep}).
The left component $g$ is also a polynomial, and by
theorem~\ref{muldeg}, has degree $2$. Solving the equation
$f=g(h)$ we have the multi--univariate decomposition:
\[(4t^2-2t+10,x+z^2y-zy^2).\]
\end{example}
\begin{example}\label{alg1-ex2}
Let $f=\displaystyle{\frac{f_n}{f_d}}$ with
\[
\begin{array}{rcl}
f_n&=&y^2x^2+2x^2yz^2-2y^6x+z^4x^2-2z^2xy^5+y^{10}\\
&&-81x^2-450xyz-625y^2z^2,\\
f_d&=&y^2x^2+2x^2yz^2-2y^6x+z^4x^2-2z^2xy^5+y^{10}\\
&&-162x^2-900xyz-1250y^2z^2.
\end{array}
\]
We look for all the intermediate fields of $\mathbb{Q}(f)\subset
\mathbb{Q}(x,y,z)$ with transcendence degree $1 $ over $\mathbb{Q}
$. First, we will factor the polynomial
\[f_n(x,y,z)\,f_d(s,t,u)-f_n(s,t,u)\,f_d(x,y,z)=-625f_1\,f_2, \]
where
\[
\begin{array}{r@{}c@{}l}
f_1&=&-xtz^2u+\frac{9}{25}xt^5-zsty-zu^2sy+zt^5y-\frac{9}{25}xz^2s\\
&&-\frac{9}{25}xu^2s-\frac{9}{25}xys-xyut-\frac{9}{25}xts+\frac{9}{25}sy^5+uty^5,\\
\\
f_2&=&-xtz^2u-\frac{9}{25}xt^5+zsty+zu^2sy-zt^5y-\frac{9}{25}xz^2s\\
&&+\frac{9}{25}xu^2s-\frac{9}{25}xys-xyut+\frac{9}{25}xts+\frac{9}{25}sy^5+uty^5.
\end{array}
\]
We have  $f_1(x,y,z,x,y,z)\not=0$, thus $f_1$ is not symmetric
near--separated. On the other hand, $f_2(x,y,z,x,y,z)=0$.
Moreover,
\[
\begin{array}{rcl}
f_2&=&-zt^5y+uty^5+\left(-\displaystyle{\frac{9}{25}}t^5-tz^2u-yut\right)x\\
&&+\left(zty+\displaystyle{\frac{9}{25}}y^5+zu^2y\right)s\\
&&+\left(-\displaystyle{\frac{9}{25}}z^2+\displaystyle{\frac{9}{25}}t+
\displaystyle{\frac{9}{25}}U^2-\displaystyle{\frac{9}{25}}y\right)sx.
\end{array}
\]
We check that $f_2$ is symmetric near--separated, by solving a
linear system of equations. Define
\[
\begin{array}{l@{\;}c@{\;}l@{\;}c@{\:}l}
f_2(x,y,z,1,0,0)&=&r(x,y,z)&=&-\displaystyle{\frac{9}{25}}xz^2-
\displaystyle{\frac{9}{25}}xy+\displaystyle{\frac{9}{25}}y^5\\
&&&=&\left(-\displaystyle{\frac{9}{25}}z^2-\displaystyle{\frac{9}{25}}y\right)x
+\displaystyle{\frac{9}{25}}y^5
\end{array}
\]
Next, we compute $s_0(y,z)$ such that
\[\displaystyle{\frac{9}{25}}y^5s_0(t,u)-\displaystyle{\frac{9}{25}}t^5s_0(y,z)
=-zt^5y+ut{y}^{5}.\] Let $s_0(y,z)=a_5(z)\,y^5+\cdots+a_0(z)$.
Then $a_1=\displaystyle{\frac{25}{9}}\,z$ and
$a_0=a_2=a_3=a_4=a_5=0$. Hence,
$s_0=\displaystyle{\frac{25}{9}}zy$ and
$s_1(y,z)=\displaystyle{\frac{r_1(y,z)\,s_0(t,u)-c_{10}}{r_0(t,u)}}=1$.
Thus $s=x+\displaystyle{\frac{25}{9}}zy$,
$s(1,0,0)=1$ and $(r,s)$
is a symmetric near--separated representation of $p$:
\[
\begin{array}{rcl}
r&=&-\displaystyle{\frac{9}{25}}xz^2-\displaystyle{\frac{9}{25}}xy+\displaystyle{\frac{9}{25}}y^5\\
s&=&x+\displaystyle{\frac{25}{9}}zy.\end{array}
\]
Now we compute $g$, which is a univariate function with degree
$2$. Solving the corresponding linear system of equations we
obtain
\[
g=\displaystyle{\frac{625t^2-6561}{625t^2-13122}}.
 \]
\end{example}

\subsection{Second algorithm}
For this algorithm, we suppose that $\mathbb{K}$ has sufficiently
many elements. If it is not the case, then we can decompose $f$ in
an algebraic extension $\mathbb{K}[\omega]$ of $\mathbb{K}$, and
then check if it is equivalent to a decomposition with
coefficients in $\mathbb{K}$; this last step can be done by
solving a system of linear equations in the same fashion as the
computation of $g$. The algorithm is based on Corollary
\ref{primeform}; we will need several technical results too.

\begin{lemma}\label{multivunits}
Let $f\in\mathbb{K}(\mathbf{x})$. Then for any admissible monomial
ordering $>$ there are units $u\in\mathbb{K}(y),\
v_i\in\mathbb{K}(x_i),\:i=1,\ldots,n$ such that, if
$\overline{f}=\overline{f}_n\,/\,\overline{f}_d=u\circ
f(v_1,\ldots,v_n)$, then $\mathrm{lm}\ \overline{f}_n>\mathrm{lm}\
\overline{f}_d$, $\overline{f}_n(0,\ldots,0)=0$ and
$\overline{f}_d(0,\ldots,0)\neq 0$.
\end{lemma}
\begin{proof}
Let $>$ be any admissible monomial ordering. Let $u_1\in
\mathbb{K}(y)$ be a unit such that $f_1=f_{1n}\,/\,f_{1d}=u_1(f)$
verifies $\mathrm{lm}\ f_{1n}>\mathrm{lm}\ f_{1d}$. Such a unit
always exists:
\begin{description}
\item [--] If $\mathrm{lm}\ f_n<\mathrm{lm}\ f_d$, let $u_1=1/y$.
\item [--] If $\mathrm{lm}\ f_n=\mathrm{lm}\ f_d$, let
$u_1=(1/y)\circ(y-\displaystyle{\frac{\mathrm{lc}\
f_n}{\mathrm{lc}\ f_d}})$.
\item [--] If $\mathrm{lm}\ f_n>\mathrm{lm}\ f_d$, let $u_1=y$.
\end{description}
Let $\underline{\alpha}=(\alpha_1,\ldots,\alpha_n)\in
\mathbb{K}^n$ such that $f_{1d}(\underline{\alpha})\neq 0$ (such a
$\alpha$ exists if $\mathbb{K}$ has sufficiently many elements).
Let $v_i=x_i+\alpha_i,\ i=1,\ldots,n$ and
$f_2=f_{2n}/f_{2d}=f_1(v_1,\ldots,v_n)$. Then
$f_{2d}(0,\ldots,0)\neq 0$, and we can take
\[u=y-\frac{f_{2n}(0,\ldots,0)}{f_{2d}(0,\ldots,0)}\] so that
$\overline{f}=u\circ f(v_1,\ldots,v_n)$ verifies all the
conditions.
\end{proof}
\begin{lemma}\label{homog}
Let $i,j,k\in\mathbb{N}$ with $i<j\leq k$,
$P,Q\in\mathbb{K}[\underline{x}]$ and $>$ an admissible monomial
ordering such that $\mathrm{lm}\ P>\mathrm{lm}\ Q$. Then
$\mathrm{lm}\ P^jQ^{k-j}>\mathrm{lm}\ P^iQ^{k-i}$.
\end{lemma}
\begin{lemma}\label{equivdec}
Let
$\overline{f}=\overline{f}_n/\overline{f}_d\in\mathbb{K}(\mathbf{x})$
such that $\mathrm{lm}\ \overline{f}_n>\mathrm{lm}\
\overline{f}_d$, $\overline{f}_n(0,\ldots,0)=0$ and
$\overline{f}_d(0,\ldots,0)\neq 0$. Then, for every
uni--multivariate decomposition $\overline{f}=g(h)$ there exists
an equivalent decomposition
$\overline{f}=\overline{g}(\overline{h})$ with
$\overline{g}=\overline{g}_n/\overline{g}_d,\mathrm{deg}\
\overline{g}_n>\mathrm{deg}\ \overline{g}_d$ and
$\overline{g}_n(0)=0$ (thus $\overline{g}_d(0)\neq 0$).
\end{lemma}
\begin{proof}
As in the proof of Lemma \ref{multivunits}, there exists a unit
$u_1$ such that if $h_1=u(h)=h_{1n}/h_{1d}$, then
$h_{1n}(0,\ldots,0)=0$. Let
$g_1=g(u^{-1})=(a_uy^u+\cdots+a_0)/(b_vy^v+\cdots+b_0)$. Then
\[\overline{f}=\frac{a_uh_{1n}^u+\cdots+a_0h_{1d}^u}{b_vh_{1n}^v
+\cdots+b_0h_{1d}^v}h_{1d}^{v-u}\] and by Corollary
\ref{primeform},
$\overline{f}_d=(b_vh_{1n}^v+\cdots+b_0h_{1d}^v)h_{1d}^{max\{u-v,0\}}$.
As $\overline{f}_d(0,\ldots,0)\neq 0$ and $h_{1n}(0,\ldots,0)=0$
we must have $h_{1d}(0,\ldots,0)\neq 0$. But
$\overline{f}_n(0,\ldots,0)=0$ and
$\overline{f}_n=(a_uh_{1n}^u+\cdots+a_0h_{1d}^u)h_{1d}^{max\{v-u,0\}}$,
thus $a_0=0$. Next, we will prove that there is an equivalent
decomposition verifying the condition on the degrees of the
left--component. To that end, we will consider three cases. Let
$>$ be any admissible monomial ordering and $w=\mathrm{deg}\
g_1=\mathrm{max}\{u,v\}$.
\begin{description}
\item [--] If $\mathrm{lm}\ h_{1n}<\mathrm{lm}\ h_{1d}$ then using
repeatedly Lemma \ref{homog}, \[\mathrm{lm}\
\overline{f}_n=\mathrm{lm}\ h_{1n}h_{1d}^{w-1}<\mathrm{lm}\
h_{1d}^w=\mathrm{lm}\ \overline{f}_d\] which contradicts our
hypothesis.
\item [--] If $\mathrm{lm}\ h_{1n}>\mathrm{lm}\ h_{1d}$, then
applying Lemma \ref{homog}, \[\begin{array}{rcl}\mathrm{lm}\
\overline{f}_n&=&\mathrm{lm}\ h_{1n}^uh_{1d}^{max\{v-u,0\}}\\
\mathrm{lm}\ \overline{f}_d&=&\mathrm{lm}
h_{1n}^vh_{1d}^{max\{u-v,0\}}.\end{array}\] As $\mathrm{lm}\
\overline{f}_n>\mathrm{lm}\ \overline{f}_d$ by hypothesis, by
Lemma \ref{homog} again we must have $u>v$, that is,
$\mathrm{deg}\ g_{1n}>\mathrm{deg}\ g_{1d}$.
\item [--] If $\mathrm{lm}\ h_{1n}=\mathrm{lm}\ h_{1d}$ then, as in
Lemma \ref{multivunits}, we can cancel the leading monomial of
$h_{1d}$ with a unit $u_2$ on the left, so that
$\overline{f}=g_2(h_2)$ with $\mathrm{lm}\ h_{2n}>\mathrm{lm}\
h_{2d}$ which is the previous case.
\end{description}
\end{proof}
Let $f=g(h)$ be a uni--multivariate decomposition of $f$ with
$f=f_n/f_d,g=(a_uy^u+\cdots+a_0)/(b_vy^v+\cdots+b_0)$ and
$h=h_n/h_d$. By the previous lemma, we can suppose $u>v$ and
$g(0)=0$, i.e. $a_0=0$. Then, as
\[f=\frac{a_uh_n^u+\cdots+a_1h_nh_d^{u-1}}{(b_vh_n^v+\cdots+b_0h_d^v)h_d^{u-v}}\]
we have that $h_n\,|\,f_n$ and $h_d\,|\,f_d$. This is the key to
the following algorithm.
\begin{algorithm}\label{alg2}
\
\begin{description}
\item [Input:] $f\in\mathbb{K}(\mathbf{x})$.
\item [Output:] $(g,h)$ a uni--multivariate decomposition of $f$, if it
exists.
\end{description}
\begin{description}
\item[A] Compute $u,v_1,\ldots,v_n$ as in Lemma \ref{multivunits}. Let
\[\overline{f}=\overline{f}_n/\overline{f}_d=u\circ
f(v_1,\ldots,v_n)\]
\item[B] Factor $\overline{f}_n$ and $\overline{f}_d$. Let
$D=\{(A_1,B_1),\ldots,(A_m,B_m)\}$ be the set of pairs $(A,B)$
such that $A,B$ divide $\overline{f}_n,\overline{f}_d$
respectively (up to product by constants). Let $i=1$.
\item[C] Check if there exists $g\in\mathbb{K}(y)$ with
$\overline{f}=g(A_i/B_i)$; if such a $g$ exists, \textbf{RETURN}
$\left(u^{-1}(g),h(v_1^{-1},\ldots,v_n^{-1})\right)$.
\item[D] If $i<m$, increase $i$ and go to \textbf{C}, otherwise
\textbf{RETURN NULL} ($f$ is uni--multivariate indecomposable).
\end{description}
\end{algorithm}
\begin{example}\label{alg2-ex1}
Let
\[\begin{array}{rcl}
f&=&4z^4y^2-8z^3y^3+8z^2yx+4z^2y^4-8zy^2x\\
&&+4x^2-2z^2y+2zy^2-2x+10
\end{array}\]
as in Example \ref{alg1-ex1}.
We take $u=t-10\in \mathbb{K}(t)$ and $v_1=x,v_2=y,v_3=z$. Then
\[\begin{array}{rcl}
\overline{f}&=&4z^4y^2-8z^3y^3+8z^2yx+4z^2y^4-8zy^2x\\
&&+4x^2-2z^2y+2zy^2-2x.
\end{array}\]
We factor $\overline{f}=2(x+z^2y-zy^2)(2x-1+2z^2y-2zy^2)$. We
first take the  candidate $(x+z^2y-zy^2)$. We have to check if
there are values of $a_1,a_2$ for which $g=a_2t^2+a_1t$ verifies
$\overline{f}=g(x+z^2y-zy^2)$. We find the solution
$a_2=4,a_1=-2$. Thus $f=(4t^2-2t+10)(x+z^2y-zy^2)$.
\end{example}
\begin{example}\label{alg2-ex2}
Let $f=\displaystyle{\frac{f_n}{f_d}}$ with
\[\begin{array}{rcl}
f_n&=&y^2x^2+2x^2yz^2-2y^6x+z^4x^2-2z^2xy^5\\
&&+y^{10}-81x^2-450xyz-625y^2z^2,\\
f_d&=&y^2x^2+2x^2yz^2-2y^6x+z^4x^2-2z^2xy^5\\
&&+y^{10}-162x^2-900xyz-1250y^2z^2,
\end{array}\] as in Example \ref{alg1-ex2}. Let $>$ be the pure
lexicographical ordering with $y>x>z$. Then $\mathrm{lm}\
f_n=\mathrm{lm}\ f_d=y^{10}$. Following the proof of lemma
\ref{multivunits}, let $u_1=1/(t-1)$, then $u_1(f)=f_{1n}/f_{1d}$
with
\[\begin{array}{rcl}
f_{1n}&=&y^2x^2+2x^2yz^2-2y^6x+z^4x^{2}-2z^2xy^5\\
&&+y^{10}-162x^2-900xyz-1250y^2z^2,\\
f_{1d}&=&81x^2+450xyz+625y^2z^2.
\end{array}\]
Now, let $\alpha=(1,0,0)$, so that the denominator of the previous
expression is non--zero at the point $\alpha$. Then
$f_{2n}/f_{2d}=u(f(x+1,y,z))$ with
\[\begin{array}{rcl}
f_{2n}&=& y^2x^2+2y^2x+y^2+2x^2yz^2+4yz^2x+2yz^2\\
&&-2y^6x-2y^6+z^4x^2+2z^4x+z^4-2z^2xy^5\\
&&-2z^2y^5+y^{10}-162x^2-324x-162\\ &&-900xyz-900yz-1250y^2z^2,\\
f_{2d}&=&81x^2+162x+81+450xyz+450yz+625y^2z^2.
\end{array}\]
As $f_{2n}(0,0,0)=-162$ and $f_{2d}(0,0,0)=81$, if $u_2=t+2$, we
have that \[u_2( u_1(
f(x+1,y,z)))=\overline{f}=\frac{\overline{f}_n}{\overline{f}_d}
\]
verifies the conditions of Lemma \ref{equivdec}. We factor
$\overline{f}_n$ and $\overline{f}_d$:
\[\begin{array}{rl}
\overline{f}_n&=\left(z^2+z^2x+y+xy-y^5\right)^2,\\
\overline{f}_d&=\left(9x+9+25yz\right )^2.
\end{array}\]
As the degree is multiplicative and $ \mathrm{deg}\ f=10$, and
also $\mathrm{lm}\ \overline{h}_n>\mathrm{lm}\ \overline{h}_d$,
the possible values of $\overline{h}_n,\overline{h}_d$ are
\[\begin{array}{rcl}
\overline{h}_n&=&z^2+z^2x+y+xy-y^5,\\
\overline{h}_d&\in&\{1,9x+9+25yz,\left(9x+9+25yz\right)^{2.}\}
\end{array}\]
To check them, let
$\overline{g}=\frac{a_2t^2+a_1t}{b_1t+b_0}$. We substitute
$\overline{h}$ in $\overline{g}$ and solve the homogeneous linear
system obtained by comparing the coefficients with those of
$\overline{f}$.
\begin{description}
\item [--] If $\overline{h}_d=1$, there is only the trivial solution,
thus $\overline{h}$ is not a candidate for $\overline{f}$.
\item [--] If $\overline{h}_d=9x+9+25yz$, we get the non--trivial solution
\linebreak $a_2=b_0=1,\:a_1=b_1=0$, thus $f$ has a
uni--multivariate decomposition
\[\begin{array}{c}
(u_1^{-1}(u_2^{-1}(\overline{g}))\:,\:\overline{h}(x-1,y,z))=({\frac
{-1+t^2}{t^2-2}},{\frac{z^2x+yx-y^5}{9x+25yz}}).
\end{array}\]
\item [--] If $\overline{h}_d=\left(9x+9+25yz\right)^2$,
the only solution is the trivial one.
\end{description}
Therefore, any uni--multivariate decomposition of $f$ is
equivalent to the decomposition $(g,h)$ computed before.
\end{example}
\section{Performance}
Both algorithms run in exponential time, since the number of
candidates to be tested is, in the worst case, exponential in the
degree of the input; the rest of the steps in both algorithms are
in polynomial time. However, in practical examples it seems that
most of the time is spent in the factorization of the associated
symmetric near--separated polynomial, in Algorithm \ref{alg1}, or
the numerator and denominator in Algorithm \ref{alg2}. We show in
Table 1 the average times obtained by running implementations of
these algorithms in Maple VI (see \cite{GRub00}) on a collection
of random functions. The parameters are the number of variables
$n$ and the degree of the rational function $d$. We have also
included the factorization times for each algorithm.
\begin{table}
\centering
\caption{Average computing times (in seconds)}
\begin{tabular}{|c c|c c|c c|}
\hline n&d&Alg \ref{alg1}&Fact.&Alg \ref{alg2}&Fact.\\
 \hline 2&10&32.17&23.15&27.03&22.44\\
2&25&68.20&46.34&51.10&40.33\\ 2&30&89.40&62.48&91.22&71.06\\
4&8&54.37&38.56&32.07&25.47\\ 4&25&89.75&65.95&64.41&46.72\\
4&30&156.87&110.30&134.60&99.87\\
8&10&234.90&162.89&156.12&116.66\\
8&25&349.44&235.41&341.11&276.85\\
8&30&654.72&454.36&678.89&511.01\\ \hline
\end{tabular}
\end{table}
\bibliographystyle{acm}
\bibliography{01}
\end{document}